\documentclass[12pt]{article}

\usepackage[width=16cm, height=22cm]{geometry}   
\usepackage[font={small}, width=14cm]{caption}
\usepackage{amsmath}
\usepackage{amssymb}
\usepackage{amsthm}
\usepackage{amscd}
\usepackage{bbm}
\usepackage{enumitem}
\usepackage{float}
\usepackage{authblk}
\usepackage{xcolor}
\usepackage{mathrsfs}
\usepackage[mathscr]{eucal}
\usepackage[normalem]{ulem}
\usepackage{graphicx}
\usepackage{tikz-cd}
\usepackage{comment}
\usepackage{pgfplots}
\usepackage{subcaption}
\usepackage{empheq}
\usepackage{fancyvrb}
\usepackage{url}

\usetikzlibrary{cd}

\newtheorem{thm}{Theorem}[section]
\newtheorem{prop}[thm]{Proposition}
\newtheorem{lem}[thm]{Lemma}

\theoremstyle{definition}
\newtheorem{dfn}[thm]{Definition}

\theoremstyle{remark}
\newtheorem{rem}[thm]{Remark}

\newtheorem{example}[thm]{Example}

\newcommand{\CC}{\mathbb{C}}
\newcommand{\NN}{\mathbb{N}}
\newcommand{\RR}{\mathbb{R}}

\newcommand{\ZZ}{\mathbb{Z}}

\title{Topological edge states of 1D chains and index theory}

\author{Guo Chuan Thiang}
\affil{\normalsize Beijing International Center for Mathematical Research, Peking University}

\date{\today}

\begin{document}
\maketitle

\begin{abstract}
We provide an elementary proof and refinement of a well-known idea from physics: a chiral-symmetric local Hamiltonian on a half-space has the same signed number of edge-localized states with energies in the bulk band gap, as its bulk winding number. The requirement of non-elementary methods to relate generic and non-generic cases is emphasized. Our hands-on approach complements a quick abstract proof based on the classical index theory of Toeplitz operators.
\end{abstract}

\section{Introduction}
The simplest example of \emph{bulk-edge correspondence} (BEC) occurs in discrete 1D model Hamiltonians, such as the Su--Schrieffer--Heeger model \cite{SSH}. A rough statement is: the chiral-symmetry condition disconnects the manifold of spectrally gapped model Hamiltonians, with winding numbers labelling the connected components; furthermore, this winding number counts the number of edge-localized zero-energy eigenstates for the Hamiltonian truncated to a half-line. The reader may recognize this claim as an index theorem,
\[
\text{Analytic index} = \text{Topological index},
\]
guaranteeing a certain stability of the above-mentioned edge states. It should come as no surprise that the above BEC is really a rephrasing of a classic index theorem, that of Toeplitz operators (see Section \ref{sec:Toeplitz}).

Still, some refinements are desirable, because the actual parameter space of models under consideration is often strictly smaller than the ``universal'' one used in the abstract index theory. Concretely, physicists often restrict to finite hopping range, or even nearest-neighbour models, and would like to understand the BEC without the need to ``borrow'' arbitrarily long-range hopping terms and/or extra internal degrees of freedom.

In the literature, there are some ``elementary'' approaches to the 1D BEC, particularly for the minimal two-band case, which are based on linear algebra and basic complex function theory, e.g.\ \cite{Chen,Mong}. In view of this, the reader might wonder what the benefits of an index-theoretic perspective are. Here are some answers:
\begin{itemize}
\item In ``elementary'' approaches, one assumes generic cases within the full parameter manifold, and apply certain ans\"{a}tze, for a simplifed analysis. There are subtleties involved in handling the submanifold of non-generic cases, particularly when the full manifold is already constrained by symmetries. See Appendix \ref{sec:non.exponential} for an example of how ans\"{a}tze can fail. Conversely, non-generic subclasses of models, such as the Su--Schrieffer--Heeger models (see Example \ref{ex:dimerized}), are used to exhibit BEC phenomena. Then extrapolation to the generic case needs justification.  ``Non-elementary'' methods such as operator/spectral perturbation theory are generally needed to resolve these issues. 
\item For finite-range Hamiltonians, the eigenvalue problem can be rephrased as a linear recurrence equation. Purely ``elementary'' methods would suggest that the BEC is a feature of the latter, independent of extra analytic structures in quantum theory. In fact, the role of self-adjointness is important, even if the spectral theory is left implicit in the use of informal terminology (e.g., ``bulk spectrum'', ``edge spectrum'', ``edge states'', ``bulk gap''). BEC seems to be restricted to settings where a good Fredholm operator theory is available.
\item BECs for differential operators (e.g.\ \cite{LT}) involves understanding the (in)sensitivity to choice of self-adjoint boundary conditions. The idea of propagating initial data subject to boundary conditions could be more closely mimicked in the case of BEC for discrete models. This is the approach that we take, and it highlights the dependence of the 1D BEC on boundary specifications.
\item The classic Toeplitz index theory and Bott periodicity proof are themselves illuminated by their connection to our ``elementary'' approach to the 1D BEC; see the proof of Theorem \ref{thm:higher.BEC} and the discussion in Section \ref{sec:Toeplitz}.
\end{itemize}

In this paper, we shall provide a general ``elementary'' approach to the 1D BEC. We devote special attention to isolating the mathematical structure essential to the phenomenon, as well as identifying points where ``non-elementary'' arguments enter. We cover the cases of arbitrary hopping range, arbitrary unit cell dimension, and possibly singular hopping terms. The main correspondence theorems appear in Section \ref{sec:chiral}. Some special versions of these results are well-known, but the precise statements, ``elementary'' proofs, and handling of non-generic cases in this paper are new.

Section \ref{sec:recurrence} introduces generalities about vector-valued recurrence equations without self-adjointness requirements. This is with a view towards ``non-Hermitian topological phases'' and other physical situations outside of a standard quantum mechanical setup. Then we study the role of self-adjointness in Section \ref{sec:SA.case}, followed by the BEC of chiral-symmetric Hamiltonians in Section \ref{sec:chiral}.

\section{Hamiltonians and vector-valued recurrence equations}\label{sec:recurrence}
\subsection{General setup and terminology}
Let $\mathcal{V}$ be a finite-dimensional complex vector space of dimension $d_\mathcal{V}$, and
\[
\mathcal{S}:=\{\psi\equiv (\psi_n)_{n\in\ZZ}\,:\,\psi_n\in \mathcal{V}\}
\]
be the linear space of $\mathcal{V}$-valued sequences. Consider the \emph{bulk parameters}
\[
V,A_r,B_r\in{\rm End}(\mathcal{V}),\quad r=1,\ldots,R.
\]
Here, $V$ is the \emph{on-site potential}, $A_r,B_r$ are respectively \emph{left-hopping} and $\emph{right-hopping}$ endomorphisms with \emph{hopping range} $r$, and $R\in \NN$ denotes the maximal hopping range. Unless otherwise stated, we assume that $A_R,B_R\in{\rm GL}(\mathcal{V})$ are invertible, but the singular case will also be handled.

We will abbreviate
\[
A:=\{A_r\}_{r=1,\ldots,R},\qquad B:=\{B_r\}_{r=1,\ldots,R}.
\]
The endomorphisms $V,A,B$ determine a linear operator $H\equiv H(V,A,B)$ on $\mathcal{S}$,
\begin{equation*}
(H\psi)_n=V\psi_n+\sum_{r=1}^R\left(B_r\psi_{n-r}+A_r\psi_{n+r}\right),\qquad n\in\ZZ,
\end{equation*}
called the \emph{bulk Hamiltonian}. For later use, we also define the \emph{Bloch Hamiltonian} at \emph{exponentiated momentum} $\lambda\in\CC^*$ to be
\begin{equation*}
\mathbf{H}(\lambda)\equiv \mathbf{H}(\lambda;V,A,B):=V+\sum_{r=1}^R\left(\lambda^{-r}B_r+\lambda^r A_r\right)\in{\rm End}(\mathcal{V}).
\end{equation*}

At each \emph{energy} $E\in\CC$, we are interested in the space of \emph{energy}-$E$ modes of $H$,
\begin{equation*}
\mathcal{M}^E\equiv \mathcal{M}^E(V,A,B):=\{\psi\in\mathcal{S}\;\,|\;\,H\psi=E\psi\}=\ker(H-E).
\end{equation*}
Let us write $V_E:=V-E$. Observe that
\begin{equation}
\psi\in\mathcal{M}^E\;\Leftrightarrow\; V_E\psi_n+\sum_{r=1}^R\left(B_r\psi_{n-r}+A_r\psi_{n+r}\right)=0,\qquad n\in\ZZ.
\label{eqn:bulk.recurrence}
\end{equation}
So $\mathcal{M}^E$ is equivalently the solution space of the order-$2R$ $\mathcal{V}$-valued linear recurrence equation, Eq. \eqref{eqn:bulk.recurrence}.

\subsection{Energy-$E$ modes from initial data}
To solve Eq.\ \eqref{eqn:bulk.recurrence}, one typically invokes the ansatz $\psi_n=\lambda^nu_\lambda, \lambda\in\CC^*$ (e.g.\ \cite{Chen,Mong}), see Eq.\ \eqref{eqn:basic.ansatz} below. Strictly speaking, this is not sufficient (see Appendix \ref{sec:exponential.validity}), and a general conceptual viewpoint is desirable.

\medskip
The invertibility of $A_R$ implies that the data in any $2R$ consecutive unit cells,
\[
\psi_n,\psi_{n+1},\ldots,\psi_{n+2R-1}\]
uniquely and linearly determines the subsequent unit cell data, $\psi_{n+2R}$. Explicitly,
\begin{equation}
\scriptsize
\underbrace{\begin{pmatrix}
0 & \mathbf{1}_\mathcal{V} & 0 & \cdots & \cdots & \cdots & \cdots & \cdots & \cdots\\
\vdots & 0 & \mathbf{1}_\mathcal{V} & \ddots & \cdots & \cdots & \cdots & \cdots & \cdots\\
\vdots & \vdots & 0 & \ddots & \ddots & \vdots & \vdots & \vdots & \vdots \\
\vdots & \vdots & \vdots & \ddots & \mathbf{1}_\mathcal{V} & 0  & \vdots & \vdots & \vdots \\
\vdots & \vdots & \vdots & \vdots & 0 & \mathbf{1}_\mathcal{V} & 0 & \vdots & \vdots \\
\vdots & \vdots & \vdots & \vdots & \vdots & 0 & \mathbf{1}_\mathcal{V} & \ddots & \vdots\\
\vdots & \vdots & \vdots & \vdots & \vdots & \vdots & \ddots & \ddots & 0\\
0 & \cdots & \cdots & \cdots & \cdots & \cdots & \cdots & 0 & \mathbf{1}_\mathcal{V}\\
-A_R^{-1}B_R & -A_R^{-1}B_{R-1} & -A_R^{-1}B_{R-2} & \cdots & -A_R^{-1} B_1 & -A_R^{-1}V_E & -A_R^{-1} A_1 & \cdots & -A_R^{-1}A_{R-1}
\end{pmatrix}}_{=: C_E\equiv C_E(V,A,B)}
\begin{pmatrix}
\psi_n\\
\psi_{n+1}\\
\vdots\\
\vdots\\
\vdots\\
\vdots\\
\vdots\\
\vdots\\
\psi_{n+2R-1}
\end{pmatrix}
=\begin{pmatrix}
\psi_{n+1}\\
\psi_{n+2}\\
\vdots\\
\vdots\\
\vdots\\
\vdots\\
\vdots\\
\vdots\\
\psi_{n+2R}
\end{pmatrix}\label{eqn:vector.companion.matrix}
\normalsize
\end{equation}
The above $2R\times 2R$ block matrix $C_E\in {\rm End}(\mathcal{V}^{\oplus 2R})$ is called the \emph{companion matrix} associated to the recurrence equation \eqref{eqn:bulk.recurrence}, and it advances data one step to the right. Up to a sign, $\det C_E=\det (A_R^{-1}B_R)$. Invertibility of $B_R$ implies that $C_E^{-1}$ exists with a similar form as $C_E$, and $C_E^{-1}$ advances data one step to the left. Here, the reader may anticipate difficulties when $A_R$ and/or $B_R$ is singular (see Example \ref{ex:dimerized}).

For convenience, we write $\mathcal{V}_n$ for the copy of $\mathcal{V}$ at the $n$-th unit cell.
The above discussion says that $\mathcal{M}^E$ is a $(2R\cdot d_\mathcal{V})$-dimensional linear space identifiable with the \emph{initial data space},
\begin{equation}
\mathcal{D}:=\mathcal{V}_1\oplus\ldots\oplus\mathcal{V}_{2R}.\label{eqn:initial.data.space}
\end{equation}
Eq.\ \eqref{eqn:bulk.recurrence} is formally solved by applying $C_E, C_E^{-1}$ iteratively to initial data in $\mathcal{D}$, thereby obtaining energy-$E$ modes in $\mathcal{M}^E$. 

The invertible companion matrix $C_E\in{\rm GL}(\mathcal{V}^{\oplus 2R})={\rm GL}(\mathcal{D})$ has a (unordered) list of \emph{nonzero generalized eigenvalues} $\lambda_i\equiv \lambda_i(E)\in\CC^*$, with respective algebraic multiplicities $m_i\geq 1$. The corresponding \emph{generalized eigenspaces}
\[
\mathcal{D}_i^E:=\{v\in \mathcal{D}\,:\,(C_E-\lambda_i)^{m_i}v=0\}
\]
are invariant subspaces for $C_E$ (see \cite{Kato}, Chapter 1, \S5.4). So the initial data space admits a canonical decomposition
\begin{equation}
\mathcal{D}=\bigoplus_i \mathcal{D}_i^E.\label{eqn:dynamical.splitting}
\end{equation}
The energy-$E$ mode space $\mathcal{M}^E\cong\mathcal{D}$ inherits a corresponding decomposition into \emph{normal mode spaces} labelled by the $\lambda_i$. \emph{Note that the ``dynamical splitting'', Eq.\ \eqref{eqn:dynamical.splitting}, may not be related to the unit cell splitting, Eq.\ \eqref{eqn:initial.data.space}.}

Depending on the size of the generalized eigenvalue $\lambda_i$, there is a further hierarchy,
\begin{equation*}
\mathcal{D}=\bigoplus_{i:|\lambda_i|<1} \mathcal{D}_i^E\bigoplus_{i:|\lambda_i|=1} \mathcal{D}_i^E\bigoplus_{i:|\lambda_i|>1}\mathcal{D}_i^E=:\mathcal{D}_{\downarrow}^E\oplus \mathcal{D}_{\rm Bloch}^E\oplus \mathcal{D}_{\uparrow}^E,
\end{equation*}
with each type of initial data generating the energy-$E$ \emph{decrease modes}, \emph{Bloch modes}, and \emph{increase modes}, respectively.

\subsubsection{Companion matrices versus Bloch Hamiltonians}
The characteristic polynomial of $C_E(V,A,B)$ can be computed as follows,
\begin{align}
P_E(\lambda)&:=\det(\lambda-C_E)\nonumber\\
&=
(-1)^{d_\mathcal{V}^2\cdot(2R-1)}\det
\footnotesize
\left(\begin{array}{ccccc|c}
-\mathbf{1} & 0 &\cdots &\cdots & \cdots & \lambda\\
\lambda & -\mathbf{1} & 0 & \cdots  & \vdots & 0 \\
0 & \ddots &\ddots &\ddots & \vdots& \vdots\\
\vdots &\ddots & \lambda & -\mathbf{1} & 0 & \vdots\\
\vdots & \cdots & 0 & \lambda & - \mathbf{1} & 0 \\
\hline
A_R^{-1}B_{R-1} & \cdots & \cdots & A_R^{-1}A_{R-2} & A_R^{-1}A_{R-1} + \lambda & A_R^{-1}B_R
\end{array}\right)
\normalsize\nonumber\\
 &= \scriptsize
 \frac{1}{\det(A_R)}\cdot\det\left(
  B_R- 
 \begin{pmatrix} B_{R-1} & \cdots & A_{R-2} & A_{R-1}+\lambda A_R\end{pmatrix}
 \left(
\begin{array}{ccccc}
-\mathbf{1} & 0 &\cdots &\cdots & \cdots \\
-\lambda & -\mathbf{1} & 0 & \cdots  & \vdots  \\
-\lambda^2 & \ddots &\ddots &\ddots & \vdots\\
\vdots &\ddots & -\lambda & -\mathbf{1} & 0 \\
-\lambda^{2R-2} & \cdots & -\lambda^2 & -\lambda & - \mathbf{1} 
\end{array}
 \right)
 \begin{pmatrix}\lambda \\ 0 \\ \vdots \\ \vdots \\ 0\end{pmatrix} 
\right)
 \normalsize\nonumber\\
&=\frac{\lambda^{R\cdot d_\mathcal{V}}}{\det(A_R)}\cdot\det\left(
\lambda^{-R}B_R+\lambda^{-(R-1)}B_{R-1}+\ldots +V_E+\ldots +\lambda^{R-1}A_{R-1}+\lambda^R A_R
\right)\nonumber\\
&=\frac{\lambda^{R\cdot d_\mathcal{V}}}{\det(A_R)}\cdot\det\left(\mathbf{H}(\lambda;V,A,B)-E\right)\label{eqn:char.poly.calculation},
\end{align}
where in the third equality, we used a standard identity for the determinant of a $2\times 2$ block matrix. Thus,
\begin{equation*}
\lambda\;\textrm{is eigenvalue of}\;C_E\quad\Leftrightarrow\quad E\;\textrm{is eigenvalue of}\;\mathbf{H}(\lambda).
\end{equation*}
The companion matrices and the Bloch Hamiltonians are complementary ways of investigating the operator $H(V,A)$. The former is useful when varying the energy $E\in \CC$, while the latter is useful when varying the (exponentiated) momentum $\lambda\in\CC^*$.

\paragraph{Exponential ansatz.} 
Suppose the initial data happens to an eigenvector of $C_E$,
\[
v_i\in\ker(C_E-\lambda_i)\subset\mathcal{D}_i^E.
\]
By considering Eq.\ \eqref{eqn:vector.companion.matrix}, it is readily seen that $v_i$ must have a nice form with respect to the unit cell splitting,
\begin{equation}
v_i=(\lambda_i u_i, \lambda_i^2 u_i,\ldots, \lambda_i^{2R}u_i),\qquad \textrm{for some}\; u_i\in\mathcal{V}.
\label{eqn:companion.to.Bloch.eigenvector}
\end{equation}
Furthermore, the last entry of Eq.\ \eqref{eqn:vector.companion.matrix} for $n=0$ becomes
\[
-B_Ru_i-B_{R-1}(\lambda_i u_i)-\ldots -V_E(\lambda_i^R u_i)-\ldots -A_{R-1}(\lambda_i^{2R-1}u_i)=A_R(\lambda_i^{2R}u_i).
\]
Equivalently, $u_i$ is an $E$-eigenvector for the Bloch Hamiltonian $\mathbf{H}(\lambda_i)$,
\begin{equation*}
0=\left(V_E+\sum_{r=1}^R\left(\lambda_i^{-r}B_r+\lambda_i^r A_r\right)\right)u_i\;\Leftrightarrow\;u_i\in\ker(\mathbf{H}(\lambda_i)-E).
\end{equation*}
Thus the $\lambda_i$-eigenvector $v_i$ of $C_E$ generates the following energy-$E$ normal mode,
\begin{equation}
\psi_n=\lambda_i^nu_i,\qquad u_i\in \ker(\mathbf{H}(\lambda_i)-E).\label{eqn:basic.ansatz}
\end{equation}
Eq.\ \eqref{eqn:basic.ansatz} is the \emph{exponential ansatz}. 

A normal mode of the form in Eq.\ \eqref{eqn:basic.ansatz} is already generated by $u_i\in\mathcal{V}$ --- the initial data can be reduced to a single unit cell $\mathcal{V}$, albeit supplemented by data of the geometric factor $\lambda_i$.

However, when $\lambda_i$ is a \emph{non-simple} eigenvalue of $C_E$, the generalized $\lambda_i$-eigenspace of $C_E$ is not easily related to the $E$-eigenspace of $\mathbf{H}(\lambda_i)$. The exponential ansatz is not applicable when the initial data $v_i$ is only a generalized eigenvector, see Appendix \ref{sec:exponential.validity} for a discussion. 

Finally, we mention that when $A_R$ and/or $B_R$ is singular, the above algebraic arguments break down, and it is even possible for normal modes (and edge states) to be compactly-supported, see Example \ref{ex:dimerized}.

\subsection{Half-space Hamiltonians and initial data for edge modes}
Define the right-half truncated sequence space,
\[
\check{S}:=\{\check{\psi}\equiv (\check{\psi}_n)_{n\geq 1}\,:\,\check{\psi}_n\in\mathcal{V}\}.
\]
The \emph{half-space Hamiltonian} $\check{H}$ is typically defined to be the truncation of the bulk Hamiltonian $H$ to $\check{\mathcal{S}}$. Explicitly,
\begin{empheq}[left={(\check{H}\check{\psi})_n=\empheqlbrace}]{align}
\sum_{r=1}^{R}B_r\check{\psi}_{n-r}+V\check{\psi}_n+\sum_{r=1}^{R}A_r\check{\psi}_{n+r},& & n\geq R+1,\nonumber\\
\sum_{r=1}^{n-1}B_r\check{\psi}_{n-r}+V\check{\psi}_{n}+\sum_{r=1}^{R}A_r\check{\psi}_{n+r},& & n=1,\ldots,R.\label{eqn:half.space.bc}
\end{empheq}
Eq.\ \eqref{eqn:half.space.bc} says that for the first $R$ unit cells, any right-hopping term arriving from $n^\prime\leq 0$ is set to zero. The \emph{half-space energy-$E$ mode spaces} are defined as
\begin{equation*}
\check{\mathcal{M}}^E:=\{\check{\psi}\in\check{\mathcal{S}}\,:\,\check{H}\check{\psi}=E\check{\psi}\}=\ker(\check{H}-E),\qquad E\in\CC.
\end{equation*}
As was the case with the bulk Hamiltonian, initial data in $\mathcal{D}=\mathcal{V}_1\oplus\ldots\oplus\mathcal{V}_{2R}$ gets propagated to the \emph{right} by the companion matrix $C_E$ to generate energy-$E$ modes. However, Eq.\ \eqref{eqn:half.space.bc} imposes extra conditions on the allowed initial data.

It is more efficient to think of elements of $\check{S}$ as sequences $(\check{\psi}_n)_{n\geq {1-R}}$ subject to the \emph{Dirichlet boundary condition} on the initial $R$ unit cells,
\begin{equation}
\check{\psi}_{1-R}=\ldots=\check{\psi}_{-1}=\check{\psi}_0=0.\label{eqn:Dirichlet.condition}
\end{equation}
Then an equivalent specification of $\check{H}$ is
\[
(\check{H}\check{\psi})_n=\sum_{r=1}^{R}B_r\check{\psi}_{n-r}+V\check{\psi}_n+\sum_{r=1}^{R}A_r\check{\psi}_{n+r}, \qquad n\geq 1,
\]
with $\check{\psi}\equiv (\check{\psi})_{n\geq 1-R}$ subject to Eq.\ \eqref{eqn:Dirichlet.condition}.
In other words, we now consider the initial data space to be
\[
\mathcal{D}=\mathcal{V}_{1-R}\oplus\ldots\oplus\mathcal{V}_0\oplus \mathcal{V}_1\oplus\ldots\oplus\mathcal{V}_{R},
\]
but restrict to the ``Dirichlet'' subspace
\[
\mathcal{D}_{\rm Dir}:=0\oplus\ldots\oplus 0\oplus\mathcal{V}_1\oplus\ldots\oplus\mathcal{V}_R\;\subset\;\mathcal{D}.
\]
Then we have the identification
\[
\check{\mathcal{M}}^E\cong\mathcal{D}_{\rm Dir}
\]
provided by the iterated application of $C_E$ on the Dirichlet initial data in $\mathcal{D}_{\rm Dir}$. Thus
\begin{equation}
\dim\check{\mathcal{M}}^E=\dim \mathcal{D}_{\rm Dir}=R\cdot d_\mathcal{V}.\label{eqn:Dirichlet.dimension}
\end{equation}

\medskip
Within $\check{M}^E\cong \mathcal{D}_{\rm Dir}$, we are interested in characterizing the subspace of so-called energy-$E$ \emph{edge} modes.
\begin{dfn}\label{dfn:edge.states}
For a Hamiltonian $H=H(V,A,B)$, an energy-$E$ edge mode $\check{\psi}$ is an eigen-solution, $\check{H}{\check{\psi}}=E\check{\psi}$ which goes to zero as $n\to \infty$.
\end{dfn}
By definition, an edge mode $\check{\psi}$ belongs to $\check{\mathcal{M}}^E\cong\mathcal{D}_{\rm Dir}$, and the decay condition just means that the initial data of $\check{\psi}$ must also come from $\mathcal{D}^E_{\downarrow}$ (i.e., initial data for decrease modes). Thus the space of energy-$E$ edge modes is precisely the linear space
\begin{equation*}
\check{\mathcal{M}}^E_{\rm edge}\cong \mathcal{D}_{\rm Dir}\cap \mathcal{D}^E_{\downarrow}.
\end{equation*}

\begin{rem}
Other ``boundary conditions'' may be imposed on the initial data space $\mathcal{D}$, by intersecting with subspaces different from $\mathcal{D}_{\rm Dir}$. In the self-adjoint case to be studied in Section \ref{sec:SA.case} later, we would also need to check the self-adjointness of such boundary conditions, analogous to the case of Schr\"{o}dinger/Dirac differential operators. So the 1D bulk-boundary correspondences that we investigate in this paper are implicitly dependent on the choice of Dirichlet boundary condition as determined by the cut-off position for the unit cells (see \cite{Thiang}, \cite{PSB} Fig.\ 2.1, and Example \ref{ex:dimerized}). This should be contrasted with the 2D BEC, which is actually robust against the choice of boundary conditions \cite{LT}.
\end{rem}

\section{Self-adjoint Hamiltonians}\label{sec:SA.case}
Now, let $\mathcal{V}$ have an inner product, then the sequence space $\mathcal{S}$ contains the Hilbert space $\ell^2(\ZZ;\mathcal{V})$ of square-summable sequences. In quantum-mechanical problems\footnote{Also classical mechanical problems with an energy conservation law, but then we would use real Hilbert spaces.}, we require $H(V,A,B)$ to be self-adjoint on $\ell^2(\ZZ;\mathcal{V})$. This is equivalent to taking
\[
V=V^*,\qquad B_r=A_r^*,\quad r=1,\ldots,R.
\]
From now on, we consider $H\equiv H(V,A)\equiv H(V,A,A^*)$ as a self-adjoint operator on $\ell^2(\ZZ;\mathcal{V})$,
\begin{equation*}
(H\psi)_n=V\psi_n+\sum_{r=1}^R\left(A_r^*\psi_{n-r}+A_r\psi_{n+r}\right),\qquad n\in\ZZ.
\end{equation*}
Similarly, the half-space Hamiltonian is the self-adjoint operator $\check{H}(V,A)$ on $\ell^2(\NN;\mathcal{V})$ given by
\[
(H\check{\psi})_n=V\check{\psi}_n+\sum_{r=1}^{R}\left(A_r^*\check{\psi}_{n-r}+A_r\check{\psi}_{n+r}\right), \qquad n\geq 1,
\]
with $\check{\psi}\in\ell^2(\NN;\mathcal{V})$ extended to $\check{\psi}=(\check{\psi}_n)_{n\geq 1-R}$ but subjected to the Dirichlet condition, Eq.\ \eqref{eqn:Dirichlet.condition}. (Here, $\NN$ denotes positive natural numbers.)

The spectrum of the self-adjoint $H$ and $\check{H}$ are closed subsets of $\RR$, and stability results of their spectra \cite{Kato} are available.

\subsection{Mode duality from self-adjointness}
Self-adjointness imposes a fundamental symmetry on the generalized eigenvalues of $C_E$.

\begin{lem}\label{lem:basic.duality}
For self-adjoint $H(V,A)$ and real energies $E\in\RR$, the companion matrix $C_E(V,A)$ has an eigenvalue $\lambda_i$ with algebraic multiplicity $m_i$ iff it has an eigenvalue $\overline{\lambda_i^{-1}}$ with algebraic multiplicity $m_i$. In this case, $E$ is an eigenvalue of the Bloch Hamiltonians $\mathbf{H}(\lambda_i)$ and $\mathbf{H}(\overline{\lambda_i^{-1}})$.
\end{lem}
\begin{proof}
Observe that
\[
(\mathbf{H}(\lambda;V,A)-E)^*=\mathbf{H}(\overline{\lambda^{-1}};V,A)-E,\qquad \lambda\in\CC^*, E\in\RR.
\]
Therefore,
\[
P_E(\lambda_i)=0\quad\Leftrightarrow\quad\det(\mathbf{H}(\lambda_i)-E)=0\quad\Leftrightarrow\quad\det(\mathbf{H}(\overline{\lambda_i^{-1}})-E)=0\quad\Leftrightarrow\quad P_E(\overline{\lambda_i^{-1}})=0.
\]
Similarly, for the first $m_i-1$ derivatives,
\[
P_E^{(j)}(\lambda_i)=0,\;\;\;j=1,\ldots, m_i-1\qquad\Leftrightarrow\qquad P_E^{(j)}(\overline{\lambda_i^{-1}})=0,\;\;\;j=1,\ldots,m_i-1.
\]
So $\lambda_i$ is an order-$m_i$ zero of $P_E(\cdot)$ iff $\overline{\lambda_i^{-1}}$ is as well. This is exactly the stated duality for the eigenvalues of $C_E$.
\end{proof}

\subsection{Real energy bands and band gap}
For self-adjoint Hamiltonians, the Bloch Hamiltonians at $\lambda=e^{ik}\in{\rm U}(1)$ are of special interest. Note that $\lambda$ is self-dual in this case. The parameter $k\in[-\pi,\pi]/_{-\pi\sim\pi}$ is called the (real) quasimomentum. For the bulk Hamiltonian $H$, only Bloch modes with $|\lambda|=1$ are bounded and contribute to the spectrum of $H$, whereas decrease modes blow up to the left, and increase modes blow up to the right.

The Bloch Hamiltonians with $\lambda=e^{ik}$ are\footnote{Another convention has $e^{-ik}$ instead of $e^{ik}$, this has a sign effect on winding numbers, Definition \ref{dfn:bulk.winding}.} Hermitian,
\[
\mathbf{H}(e^{ik};V,A)=\mathbf{H}(e^{ik};V,A)^*,
\]
and each of them has $d_\mathcal{V}$ real eigenvalues (counted with multiplicities), labelled in increasing order as
\[
E_j(k),\qquad j=1,\ldots, d_\mathcal{V},\quad e^{ik}\in{\rm U}(1).
\]
The $E_j$ define the energy band functions on the \emph{Brillouin zone} ${\rm U}(1)$. They are continuously defined, although there is the complicated issue of smoothness at band crossings (when $E_j(k)=E_{j+1}(k)$ occurs), which we will not address. What is important is the presence of a \emph{band gap}.
\begin{dfn}
The self-adjoint Hamiltonian $H(V,A)$ is \emph{gapped} if there is a $j$ such that
\[
E_-:=\sup_{e^{ik}\in{\rm U}(1)}E_j(k)\;<\;\inf_{e^{ik}\in{\rm U}(1)}E_{j+1}(k)=:E_+.
\]
In this case, the interval $(E_-,E_+)$ is called the \emph{band gap}.
\end{dfn}

\begin{prop}\label{prop:dim.dec.modes}
Let $H$ be a gapped self-adjoint Hamiltonian with band gap $(E_-,E_+)$. For any $E\in(E_-,E_+)$, we have
\[
\dim\mathcal{D}_{\downarrow}^E=\dim\mathcal{D}_{\uparrow}^E=R\cdot d_\mathcal{V}.
\]
\end{prop}
\begin{proof}
Since $E\in(E_-,E_+)$, for each $e^{ik}\in{\rm U}(1)$, the Bloch Hamiltonian $\mathbf{H}(e^{ik})$ does not have eigenvalue $E$. By (the proof of) Lemma \ref{lem:basic.duality}, $C_E$ does not have eigenvalues on the unit circle. By the same Lemma, we have
\[
\dim\mathcal{D}_{\downarrow}^E=\dim\mathcal{D}_{\uparrow}^E=\frac{1}{2}\dim(\mathcal{D})=R\cdot d_\mathcal{V}.
\] 
\end{proof}

\subsubsection{Isolated edge states in band gap}
The total spectrum $\sigma(H)$ of $H(V,A)$ is obtained as the union of the ranges of the energy band functions $E_j$. It is also called the \emph{bulk spectrum} of $H(V,A)$. If $A_R$ is invertible, $H(V,A)$ has no point spectrum (see Appendix \ref{sec:exponential.validity}; whether or not $|\lambda_i|=1$, there are no bounded modes in $\mathcal{M}^E$). 

In general, the spectrum of the half-space Hamiltonian $\check{H}(V,A)$ and that of $H(V,A)$ are different. For example, it is certainly possible that $\check{H}(V,A)$ acquires point spectrum at certain eigenvalues $E$. These point spectra are due to the edge states of Definition \ref{dfn:edge.states}; note that the latter are square-summable, see Appendix \ref{sec:exponential.validity}.

A rather non-trivial spectral theory result is that the \emph{essential spectrum} of $\check{H}$ and $H$ coincide. For example, one might construct Weyl sequences from the Bloch modes, supported on the right-half Hilbert space. Alternatively, one can invoke Toeplitz operator theory, see Chapter 4.3 of \cite{Arveson}. Thus $\sigma(\check{H})\cap (E_-,E_+)$ can only comprise \emph{discrete spectrum}, i.e., isolated eigenvalues with finite multiplicity.

\begin{rem}
Prop.\ \ref{prop:dim.dec.modes} and Eq.\ \eqref{eqn:Dirichlet.dimension} together show that for each $E\in(E_-,E_+)$, the space of energy-$E$ edge modes is the intersection of two linear manifolds, each of real dimension $2R\cdot d_\mathcal{V}$. One of them, $\mathcal{D}_{\rm Dir}$, is fixed, while the other, $\mathcal{D}_\downarrow^E$, depends on $(V,A)$ and $E$. The intersection takes place inside the real $(4R\cdot d_\mathcal{V})$-dimensional initial data space $\mathcal{D}$. For generic choices of $(V,A;E)$, there will be no non-trivial edge modes.
\end{rem}

\begin{rem}\label{eqn:meaning.deforming.Hamiltonians}
A priori, the ``space of model Hamiltonians'' is topologized by the matrix norms of $V,A$. Indeed, we will define homotopy invariants based upon this parameter space. At the same time, $H(V,A)$ unitarily Fourier transforms into the multiplication operator by the Bloch Hamiltonian function,
\[
{\rm U}(1):e^{ik}\mapsto \mathbf{H}(e^{ik};V,A),
\]
acting on $L^2({\rm U}(1);\mathcal{V})$.
In this Fourier transformed description, the operator norm is given by ${\rm sup}_{e^{ik}\in{\rm U}(1)}||\mathbf{H}(e^{ik};V,A)||$, and this is continuous in the parameters $(V,A)$. The half-space Hamiltonian is obtained as the composition
\[
\check{H}(V,A)=p\circ H(V,A)\circ\iota,
\]
where $\iota:\ell^2(\NN;\mathcal{V})\hookrightarrow\ell^2(\ZZ;\mathcal{V})$ is the inclusion and $p=\iota^*$ is the projection. 

Therefore, the operator norms of $H(V,A)$ and $\check{H}(V,A)$ are continuously controlled by the bulk parameters $(V,A)$. This is relevant for the application of spectral perturbation theory. For example, the in-gap eigenvalues of $\check{H}(V,A)$ will vary continuously with $(V,A)$, see \cite{Kato} \S3.5. The parameters $(V,A)$ are allowed to vary a lot, so edge states are generally not very robust. An exception occurs when $(V,A)$ are constrained by a chiral symmetry, as we shall see.
\end{rem}

\section{Chiral-symmetric Hamiltonians}\label{sec:chiral}

\begin{dfn}
A \emph{chiral symmetry operator} is a grading operator $\Gamma=\Gamma^*=\Gamma^{-1}$ on $\mathcal{V}$. It extends in the obvious way to a grading operator on $\ell^2(\ZZ;\mathcal{V})$. A (self-adjoint) Hamiltonian $H=H(V,A)$ is \emph{chiral-symmetric} if $H\Gamma=-\Gamma H$.
\end{dfn}

A chiral symmetry operator satisfies $\Gamma^2=1$, and it orthogonally decomposes the unit cell Hilbert space into
\[
\mathcal{V}=\mathcal{V}_+\oplus \mathcal{V}_-,
\]
according to its $\pm$ eigenspaces, called the $\pm$-\emph{graded components}. Similarly, $\ell^2(\ZZ;\mathcal{V})\cong\ell^2(\ZZ;\mathcal{V}_+)\oplus \ell^2(\ZZ;\mathcal{V}_-)$ is split into ``sublattice'' Hilbert spaces. In any basis adapted to this splitting, chiral symmetry of $H$ is equivalent to $V$ and $A=\{A_1,\ldots,A_R\}$ being off-diagonal matrices,
\[
\Gamma=\begin{pmatrix} \mathbf{1}_{\mathcal{V}_+} & 0 \\ 0 & -\mathbf{1}_{\mathcal{V}_-} \end{pmatrix},\qquad V=\begin{pmatrix} 0 & v^*\\ v & 0\end{pmatrix},\qquad A_r=\begin{pmatrix}0 & a_{r,-+}\\
a_{r,+-} & 0\end{pmatrix},
\]
where $v,a_{r,+-}:\mathcal{V}_{+}\to\mathcal{V}_{-}$ and $a_{r,-+}:\mathcal{V}_-\to\mathcal{V}_+$. We also have $a_{R,+-}, a_{R,-+}$ invertible iff $A_R$ is invertible. Consequently, the Bloch Hamiltonians are off-diagonal operators,
\begin{align}
\mathbf{H}(\lambda)&=\begin{pmatrix}
0 & v^*+\sum_{r=1}^R\left(\lambda^{-r}a^*_{r,+-}+\lambda^r a_{r,-+}\right)\\
v+\sum_{r=1}^R\left(\lambda^{-r}a^*_{r,-+}+\lambda^r a_{r,+-}\right) & 0
\end{pmatrix}\nonumber\\
&=:\begin{pmatrix} 0 & h_{-+}(\lambda)\\
h_{+-}(\lambda) & 0\end{pmatrix},\label{eqn:chiral.Bloch.Hamiltonian}
\end{align}
and $\Gamma \mathbf{H}(\lambda)\Gamma=-\mathbf{H}(\lambda)$ holds. 

When restricted to $\lambda=e^{ik}\in{\rm U}(1)$, the Hermitian Bloch Hamiltonians $\mathbf{H}(e^{ik})$ have the property that $h_{-+}(e^{ik})=h_{+-}^*(e^{ik})$. Therefore
\begin{align*}
H\; \textrm{is gapped}\;&\Leftrightarrow\;\mathbf{H}(e^{ik})\;\textrm{invertible} &&\forall\,e^{ik}\in{\rm U}(1),\\
&\Leftrightarrow\; \mathbf{H}(e^{ik})^2=\begin{pmatrix} h_{+-}^*(e^{ik})h_{+-}(e^{ik}) & 0 \\ 0 & h_{+-}(e^{ik})h_{+-}^*(e^{ik})\end{pmatrix} >0 &\;\;&\forall\,e^{ik}\in{\rm U}(1),\\
&\Leftrightarrow\; h_{+-}(e^{ik})\;\;\textrm{invertible} && \forall\,e^{ik}\in{\rm U}(1).
\end{align*}
Then the following is well-defined:
\begin{dfn}\label{dfn:bulk.winding}
Let $H=H(V,A)$ be a chiral-symmetric gapped Hamiltonian, with $A_R$ possibly being singular. Its \emph{bulk winding number} is
\[
\mathscr{W}(H):={\rm Wind}\left(\det(h_{+-}|_{{\rm U}(1)})\right)=-{\rm Wind}\left(\det(h_{-+}|_{{\rm U}(1)})\right)\in\ZZ.
\]
\end{dfn}
Here, we recall that the winding number of a differentiable curve $\gamma:{\rm U}(1)\to \CC^*$ about the origin can be defined via the contour integral
\[
{\rm Wind}(\gamma)=\frac{1}{2\pi i}\oint_{\gamma} \frac{dz}{z}.
\]

Note that a continuous change of $(V,A)$, thus of $H(V,A)$ by Remark \ref{eqn:meaning.deforming.Hamiltonians}, induces a homotopy of the map $\det h_{+-}:{\rm U}(1)\to\CC^*$. So Definition \ref{dfn:bulk.winding} is a \emph{topological invariant} of $H(V,A)$. We also mention that traditional solid-state physics constructions such as Berry connection on the Bloch bundle, Berry/geometric phase etc., are not needed to define $\mathcal{W}(H)$.

Next, we examine the half-space Hamiltonian $\check{H}=\check{H}(V,A)$, which is also of the off-diagonal form,
\[
\check{H}=\begin{pmatrix}
0 & \check{H}_{-+}\\
\check{H}_{+-} & 0
\end{pmatrix},\qquad \check{H}_{-+}=\check{H}_{+-}^*:\ell^2(\NN;\mathcal{V}_-)\to\ell^2(\NN;\mathcal{V}_+).
\]
The zero-energy mode space, i.e., $\ker\check{H}$, is special because it admits a splitting according to the $\Gamma$-grading,
\[
\check{\mathcal{M}}_{\rm edge}^0\equiv \ker\check{H}=\ker\check{H}_{+-}\oplus\ker\check{H}_{-+}=:\check{\mathcal{M}}^0_{{\rm edge},+}\oplus\check{\mathcal{M}}^0_{{\rm edge},-}.
\]
\begin{dfn}\label{dfn:edge.index}
Let $H=H(V,A)$ be a chiral-symmetric gapped Hamiltonian, with $A_R$ possibly being singular. Its \emph{edge index} is defined to be
\[
{\rm Ind}_e(\check{H}):=\dim\ker\check{H}_{+-}-\dim\ker\check{H}_{-+}\in\ZZ.
\]
\end{dfn}
The \emph{bulk-edge correspondence} refers to claims of the form
\[
\mathcal{W}(H)\overset{?}{=}{\rm Ind}_e(\check{H}).
\]

\subsection{Two-band nearest-neighbour case}
The base case for the bulk-edge correspondence has $R=1,d_\mathcal{V}=2$, where it is the sharpest.

\begin{thm}\label{thm:2.band.BEC}
Let $H(V,A)$ be a chiral-symmetric gapped Hamiltonian with $R=1, d_\mathcal{V}=2$, and $A=A_{R=1}$ possibly singular. Then $\mathcal{W}(H)\in\{-1,0,+1\}$, and
\[
\begin{cases}
\dim\ker\check{H}_{+-}=1,\;\;\;\,\dim\ker\check{H}_{-+}=0,\quad & \mathcal{W}(H)=1,\\
\dim\ker\check{H}_{+-}=0,\;\;\;\,\dim\ker\check{H}_{-+}=0,\quad & \mathcal{W}(H)=0,\\
\dim\ker\check{H}_{+-}=0,\;\;\;\,\dim\ker\check{H}_{-+}=1,\quad & \mathcal{W}(H)=-1.
\end{cases}
\]
In particular,
\begin{equation*}
{\rm Ind}_e(\check{H})\equiv\dim\ker\check{H}_{+-}-\dim\ker\check{H}_{-+}=\mathcal{W}(H).
\end{equation*}
\end{thm}
\begin{proof}
First, we assume that $A$ is invertible, and $C_0(V,A)$ has only simple eigenvalues. The proof in this case is indeed ``elementary'':

Let the (simple) eigenvalues of $C_0$ lying inside the unit circle be $\lambda_i, i=1,2$, with corresponding eigenvectors $v_i=(\lambda_iu_i,\lambda_i^{2}u_i)$, where $u_i\in{\rm ker}(\mathbf{H}(\lambda_i))$. (The exponential ansatz, Eq.\ \eqref{eqn:companion.to.Bloch.eigenvector}, is valid). Because of the form, Eq.\ \eqref{eqn:chiral.Bloch.Hamiltonian}, of $\mathbf{H}(\lambda_i)$, its kernel splits as
\[
\underbrace{\ker(\mathbf{H}(\lambda_i))}_{\dim 1}=\underbrace{\ker h_{+-}(\lambda_i)}_{\subset\,\mathcal{V}_+}\oplus \underbrace{\ker h_{-+}(\lambda_i)}_{\subset\,\mathcal{V}_-},
\]
and $u_i$ lies in exactly one of these two components.
Let us write
\[
I_\pm=\{i:u_i\in\mathcal{V}_\pm\}.
\]
Then
\begin{equation}
u_i\propto \begin{cases}\binom{1}{0}, & i\in I_+,\\
\binom{0}{1}, & i\in I_-,
\end{cases}\label{eqn:2.band.chiral.eigenvectors}
\end{equation}
and the zero-energy decrease mode space splits as
\[
\mathcal{D}^0_{\downarrow}=\mathcal{D}^0_{\downarrow,+}\oplus\mathcal{D}^0_{\downarrow,-},
\]
where
\[
\mathcal{D}^0_{\downarrow,+}={\rm span}\left\{v_i=\begin{pmatrix}
\binom{\lambda_i}{0}\\
\binom{\lambda_i^{2}}{0}
\end{pmatrix}
: i\in I_+\right\},\qquad\mathcal{D}^0_{\downarrow,-}={\rm span}\left\{v_i=\begin{pmatrix}
\binom{0}{\lambda_i}\\
\binom{0}{\lambda_i^{2}}
\end{pmatrix}
: i\in I_-\right\}.
\]
Now we impose the Dirichlet condition on the initial unit cell, to obtain the admissible edge states. Since the $\lambda_i$ are distinct, the Dirichlet condition imposes exactly $R=1$ constraint on $\mathcal{D}^0_{\downarrow,+}$ (and similarly on $\mathcal{D}^0_{\downarrow,-}$), thus
\begin{align}
\dim \ker \check{H}_{+-}&\equiv\dim \check{\mathcal{M}}^0_{{\rm edge},+}\cong \dim \mathcal{D}^0_{\downarrow,+}\cap\mathcal{D}_{\rm Dir}=\max\{0,|I_+|-1\},\label{eqn:even.kernel}\\
\dim \ker \check{H}_{-+}&\equiv\dim \check{\mathcal{M}}^0_{{\rm edge},-}\cong \dim \mathcal{D}^0_{\downarrow,-}\cap\mathcal{D}_{\rm Dir}=\max\{0,|I_-|-1\}.\label{eqn:odd.kernel}
\end{align}

As for the winding number, note that $h_{+-}$ is a Laurent polynomial, of the form
\begin{equation}
h_{+-}=\overline{a_{-+}}\lambda^{-1}+v+a_{+-}\lambda,\qquad v\in\CC,\; a_{+-},a_{-+}\in\CC^*.\label{eqn:2-band.chiral.Laurent}
\end{equation}
Its winding number lies between $-1$ and $+1$. By the argument principle,
\begin{align}
\mathcal{W}(H)={\rm Wind}\, h_{+-}|_{{\rm U}(1)}&=\# \textrm{zeroes of}\; h_{+-}\; \textrm{inside unit circle}\,-\,\#\text{\rm poles of}\; h_{+-}\; \textrm{inside unit circle}\nonumber\\
&=|I_+| - 1,\nonumber\\
-\mathcal{W}(H)={\rm Wind}\, h_{-+}|_{{\rm U}(1)}&=\# \textrm{zeroes of}\; h_{-+}\; \textrm{inside unit circle}\,-\,\#\text{\rm poles of}\; h_{-+}\; \textrm{inside unit circle}\nonumber\\
&=|I_-| - 1.\label{eqn:argument.principle}
\end{align}
Eq.\ \eqref{eqn:even.kernel}--\eqref{eqn:odd.kernel}, together with Eq.\ \eqref{eqn:argument.principle}, give the result.

For the general case where $A$ is singular and/or $C_0(V,A)$ has non-simple eigenvalues, the argument is not so straightforward. We may approximate $(V,A)$ by a neighbouring $(V^\prime,A^\prime)$ to reach the previous case. However, such a perturbation will relate the spectra of $\check{H}(V,A)$ and $\check{H}(V^\prime,A^\prime)$ lying in a \emph{neighbourhood} of $E=0$. Eigenvalues may flow in/out of the zero-energy level as a result of the approximation. Fortunately, Prop.\ \ref{prop:edge.modes.energy.zero}, proved later, guarantees that the in-gap spectrum can only occur at $E=0$.
\end{proof}

For later use, we highlight three examples with singular $A$, which all other chiral-symmetric gapped Hamiltonians will be deformed to, as explained in Section \ref{sec:more.bands}.
See Appendix \ref{sec:non.exponential} for a family of examples whose $C_0(V,A)$ have non-simple eigenvalues.
\begin{example}\label{ex:dimerized} Consider
\[
V=0,\qquad A=\begin{pmatrix} 0 & 0 \\ 1 & 0\end{pmatrix},
\]
so $H(V,A)$ is just the operator which hops between the $\mathcal{V}_+$ and $\mathcal{V}_-$ of \emph{adjacent} unit cells. $\check{H}(V,A)$ is the same, except for a ``cut-off'' at the boundary unit cell,
\begin{center}
 \begin{tikzpicture}
\node(B-1)[gray] at (-1,0) {$\mathcal{V}_-$};
\node(A0) at (0,0) {$\mathcal{V}_+$};
\node(B0) at (1,0) {$\mathcal{V}_-$};
\node(A1) at (2,0) {$\mathcal{V}_+$};
\node(B1) at (3,0) {$\mathcal{V}_-$};
\node(A2) at (4,0) {$\mathcal{V}_+$};
\node at (-0.5,0) {$||$};
\node at (1.5,0) {$|$};
\node at (3.5,0) {$|$};
\node at (5,0) {$\cdots$};
\draw[->,gray] (A0) to [out=105, in=75] (B-1) node at (-0.5,0.9) {$A$};
\draw[->] (A1) to [out=105, in=75] (B0) node at (1.5,0.9) {$A$};
\draw[->] (A2) to [out=105, in=75] (B1) node at (3.5,0.9) {$A$};
 \end{tikzpicture}
\end{center}
Other than the boundary unit cell, each adjacent $\mathcal{V}_-, \mathcal{V}_+$ pair is coupled by the Hermitian matrix $\begin{pmatrix} 0 & 1 \\ 1 & 0 \end{pmatrix}$, which has eigenvalues $\pm 1$. At the boundary unit cell, the left-hopping term $A$ is set to zero, so there is a zero eigenvalue. In total, $\check{H}(V,A)$ has spectrum $\{-1,0,+1\}$, with $\{-1,+1\}$ being infinitely-degenerate (thus essential spectrum\footnote{This is an example of ``flat band spectrum''.}), while $\{0\}$ is discrete spectrum arising from the \emph{compactly-supported} edge state 
\[
\left(\binom{1}{0},\binom{0}{0},\binom{0}{0},\ldots\right).
\]
So ${\rm Ind}_e(\check{H})=1$. The Bloch Hamiltonians are
\[
\check{H}(\lambda)=\begin{pmatrix} 0 & \lambda^{-1} \\ \lambda & 0\end{pmatrix},
\]
and so $\mathcal{W}(H)=+1$. These calculations are consistent with Theorem \ref{thm:2.band.BEC}. 

A winding number $-1$ example is obtained by swapping the roles of $\mathcal{V}_-$ and $\mathcal{V}_+$. 

The basic ``trivial'' winding 0 example is given by
\[
V=\begin{pmatrix} 0 & 1 \\ 1 & 0 \end{pmatrix},\qquad A=0,
\]
and it pairs up adjacent $\mathcal{V}_+, \mathcal{V}_-$ within the same unit cell:
\begin{center}
 \begin{tikzpicture}
\node(B-1)[gray] at (-1,0) {$\mathcal{V}_-$};
\node(A0) at (0,0) {$\mathcal{V}_+$};
\node(B0) at (1,0) {$\mathcal{V}_-$};
\node(A1) at (2,0) {$\mathcal{V}_+$};
\node(B1) at (3,0) {$\mathcal{V}_-$};
\node(A2) at (4,0) {$\mathcal{V}_+$};
\node at (-0.5,0) {$||$};
\node at (1.5,0) {$|$};
\node at (3.5,0) {$|$};
\node at (5,0) {$\cdots$};
\draw[->] (B0) to [out=105, in=75] (A0) node at (0.5,0.9) {$V$};
\draw[->] (B1) to [out=105, in=75] (A1) node at (2.5,0.9) {$V$};
 \end{tikzpicture}
\end{center}
This $\check{H}(V,A)$ clearly has spectrum $\{-1,+1\}$ with no edge states at all.

Notice that the bulk Hamiltonians in these three examples are unitarily related to each other by a change of convention in the unit cell labelling and/or grading $\mathcal{V}_+\leftrightarrow\mathcal{V}_-$. These conventions are implicitly specified via boundary conditions, and the intrinsic meaning of the ``bulk winding number invariant'' is actually quite subtle, see \cite{Thiang,TZ}.

The above examples are ``dimerized limits'' inside the (singular) subclass of Su--Schrieffer--Heeger models \cite{SSH},
\[
V=\begin{pmatrix} 0 & t_1 \\ t_1 & 0 \end{pmatrix},\qquad A=\begin{pmatrix} 0 & 0 \\ t_2 & 0 \end{pmatrix},\qquad t_1, t_2 \in\RR.
\]
Here, the winding number is $+1$ when $|t_2|>|t_1|$, and $0$ when $|t_2|<|t_1|$, with a ``topological phase transition'' at $|t_1|=|t_2|$. For $|t_2|>|t_1|>0$, it is easily checked that the edge state is
\[
\check{\psi}_n=\left(-t_1/t_2\right)^n\binom{1}{0},\qquad n\geq 1.
\]
Note that this edge state does \emph{not} satisfy the Dirichlet condition at $n=0$ (instead, $\check{\psi}_0\in\ker A^*$). It is also \emph{not} a superposition of two normal modes with different decay rates.
\end{example}

\begin{prop}\label{prop:edge.modes.energy.zero}
Let $H(A,V)$ be a chiral-symmetric gapped Hamiltonian with $R=1,d_\mathcal{V}=2$, and $A_{R=1}$ possibly singular. For all non-zero $E$ in the band gap $(E_-,E_+)$, there are no energy-$E$ edge modes.
\end{prop}
\begin{proof}
We consider invertible $A=A_1$ first. Let $E\in(E_-,E_+)$ with $E\neq 0$. By Prop.\ \ref{prop:dim.dec.modes}, $C_{E}(V,A)$ has two eigenvalues in the unit circle, counted with multiplicity. 
\begin{itemize}
\item 
We first assume that the eigenvalues $\lambda_1,\lambda_2$ of $C_{E}(V,A)$ lying in the unit circle are distinct (thus simple). Choose a corresponding eigen-basis $\{v_1,v_2\}$ for $\mathcal{D}_\downarrow^{E}$. Recall from Eq.\ \eqref{eqn:companion.to.Bloch.eigenvector} that $v_i=(\lambda_iu_i,\lambda_i^{2}u_i)$ with $u_i\in\ker(\mathbf{H}(\lambda_i)-E)\subset\mathcal{V}$. 
Energy-$E$ edge states belong to the intersection
\[
\check{\mathcal{M}}^E_{\rm edge}=\mathcal{D}_\downarrow^{E}\cap\mathcal{D}_{\rm Dir}={\rm span}\left\{\begin{pmatrix}\lambda_1u_1\\\lambda_1^2u_1\end{pmatrix},\begin{pmatrix}\lambda_2u_2\\\lambda_2^2u_2\end{pmatrix}\right\}\cap\left\{v=\binom{0}{*}\in\mathcal{V}\oplus\mathcal{V}\right\}.
\]
Thus $\check{\mathcal{M}}^E_{\rm edge}$ is non-trivial iff $u_1\parallel u_2$. In this case, both $\mathbf{H}(\lambda_1)-E$ and $\mathbf{H}(\lambda_2)-E$ annihilate $u_1$. By conjugating with $\Gamma$, we also see that both $\mathbf{H}(\lambda_1)+E$ and $\mathbf{H}(\lambda_2)+E$ annihilate $\Gamma u_1$. Since $\mathbf{H}(\lambda_1), \mathbf{H}(\lambda_2)$ have the same two (distinct) eigenvalues, $+E$ and $-E$, and the same eigenspaces, they must be equal.

Recalling the form of chiral-symmetric Bloch Hamiltonians in Eq.\ \eqref{eqn:chiral.Bloch.Hamiltonian}, \eqref{eqn:2-band.chiral.Laurent}, equality of the $h_{+-}(\lambda_i)$ part of $\mathbf{H}(\lambda_i)$ means
\[
\lambda_1^{-1}\overline{a_{-+}}+v+\lambda_1 a_{+-}=h_{+-}(\lambda_1)=h_{+-}(\lambda_2)=\lambda_2^{-1}\overline{a_{-+}}+v+\lambda_2 a_{+-},
\]
and a little algebra leads to
\[
\overline{a_{-+}}=\lambda_1\lambda_2 a_{+-}.
\]
Simultaneously, equality of the $h_{-+}$ part means
\[
\lambda_1^{-1}\overline{a_{+-}}+\overline{v}+\lambda_1 a_{-+}=h_{-+}(\lambda_1)=h_{-+}(\lambda_2)=\lambda_2^{-1}\overline{a_{+-}}+\overline{v}+\lambda_2 a_{-+},
\]
so
\[
\overline{a_{+-}}=\lambda_1\lambda_2 a_{-+}.
\]
But $|\lambda_1|,|\lambda_2|<1$, so we have a contradiction.

\item Next, suppose $C_{E}=C_{E}(V,A)$ has a non-simple eigenvalue $\lambda$ with $|\lambda|<1$ (which is possibly defective). Since $\mathcal{D}^E_\downarrow$ is two-dimensional (Prop.\ \ref{prop:dim.dec.modes}), it is exactly the generalized $\lambda$-eigenspace of $C_E$. To get an energy-$E$ edge mode, we need to use Dirichlet initial data, $v\in \mathcal{D}_{\rm Dir}\cap \mathcal{D}^E_\downarrow$, that is, $v=\binom{0}{u}$ for some $u\in\mathcal{V}$. 

The generalized eigenvector condition on $v$ is
\begin{align*}
0&=(C_E-\lambda)^2v\\
&=(C_E-\lambda)\begin{pmatrix}
-\lambda & \mathbf{1}_\mathcal{V}\\
-A^{-1}A^* & -\lambda-A^{-1}V_E
\end{pmatrix}\binom{0}{u}\\
&=(C_E-\lambda)\binom{u}{-(\lambda+A^{-1}V_E) u}.
\end{align*}
Since $\binom{u}{-(\lambda+A^{-1}V_E) u}$ is a genuine $\lambda$-eigenvector of $C_E$, Eq.\ \eqref{eqn:companion.to.Bloch.eigenvector} says that it must equal $\binom{u}{\lambda u}$ where
\[
u\in\ker(\mathbf{H}(\lambda)-E)=\ker(\lambda^{-1}A^*+V_E+\lambda A);
\]
Concurrently, $0=\lambda u+(\lambda+A^{-1}V_E)u=(2\lambda+A^{-1}V_E)u=0$ holds. So $u$ lies in the joint kernel of two operators, thus it also lies in the kernel of their difference,
\begin{align*}
u&\in\ker(\lambda^{-1}A^*+V_E+\lambda A)\cap \ker(2\lambda A+V_E)\\
&\Rightarrow u\in\ker(\lambda A-\lambda^{-1}A^*).
\end{align*}
Now, $\lambda A-\lambda^{-1}A^*\neq 0$, since otherwise,
\[
\lambda A=\lambda^{-1}A^*\;\Rightarrow\;|\lambda|^2|\det(A)|=|\lambda|^{-2}|\det A^*|\;\Rightarrow |\lambda|=1,
\]
contradicting $|\lambda|<1$. Also, $\lambda A-\lambda^{-1}A^*$ is off-diagonal, so its kernel is spanned by either $\binom{1}{0}$ or $\binom{0}{1}$. In the first case, $\binom{1}{0}\parallel u\in\ker(2\lambda A+V_E)$ requires
\[
\binom{0}{0}=(2\lambda A+V_E)\binom{1}{0}=\begin{pmatrix} -E & * \\ * & -E\end{pmatrix}\binom{1}{0}=\binom{-E}{*},
\]
contradicting $E\neq 0$. Similarly for the second case.
\end{itemize}
Finally, suppose $H(V,A)$ has singular $A$. We can approximate $H(V,A)$ arbitrarily well by some $H(V^\prime,A^\prime)$ with invertible $A^\prime$. For the latter, we had already excluded in-gap spectrum away from zero, so the same holds for $H(V,A)$.

\end{proof}

\begin{rem}
Prop.\ \ref{prop:edge.modes.energy.zero} has a generalization to the class of \emph{Dirac-type Hamiltonians}; see \cite{Mong}, which does not, however, address non-simple eigenvalues of $C_E$ or singular $A$.
\end{rem}

\subsection{More bands and longer range models}\label{sec:more.bands}
In the general case, $R\geq 1, d_\mathcal{V}\geq 2$, a chiral-symmetric $H(V,A)$ has Bloch Hamiltonian with off-diagonal term $h_{+-}$ being an ${\rm End}(\mathcal{V}_+,\mathcal{V}_-)$-valued Laurent polynomial with zeroes/poles of highest order $R$. 

Suppose $A_R$ is invertible. Then $\det h_{+-}$ is a Laurent polynomial, whose highest-order zeroes/poles have degree $R\cdot d_\mathcal{V}/2$. The argument principle calculation, Eq.\ \eqref{eqn:argument.principle}, now gives
\[
\mathcal{W}(H)=|I_+|-\frac{R\cdot d_\mathcal{V}}{2}=\frac{R\cdot d_\mathcal{V}}{2}-|I_-|.
\]
However, instead of Eq.\ \eqref{eqn:2.band.chiral.eigenvectors}, $u_i\in\ker(\mathbf{H}(\lambda_i))$ now only gives
\[
i\in I_+\;\Rightarrow\; u_i\propto\binom{w_i}{0},\quad w_i\in\mathcal{V}_+,\qquad\;\;\; i\in I_-\;\Rightarrow\; u_i\propto\binom{0}{w_i},\quad w_i\in\mathcal{V}_-.
\]
Thus
\[
\mathcal{D}^0_{\downarrow,+}={\rm span}\left\{\begin{pmatrix}
\lambda_i\binom{w_i}{0}\\
\vdots\\
\lambda_i^{2R}\binom{w_i}{0}
\end{pmatrix}:i\in I_+\right\},\qquad w_i\in\mathcal{V}_+,
\]
and similarly for $\mathcal{D}^0_{\downarrow,-}$.
The Dirichlet condition on the initial $R$ unit cells ``generically'' imposes the maximal number, $R\cdot d_\mathcal{V}/2$, of constraints, but possibly \emph{fewer}. So instead of Eq.\ \eqref{eqn:even.kernel}--\eqref{eqn:odd.kernel}, we can only say that
\begin{align*}
|I_+|\geq\dim\ker \check{H}_{+-}&\geq\max\{0,|I_+|-R\cdot d_\mathcal{V}/2\}=\max \{0,\mathcal{W}(H)\},\\
|I_-|\geq\dim\ker \check{H}_{-+}&\geq\max\{0,|I_-|-R\cdot d_\mathcal{V}/2\}=\max \{0,-\mathcal{W}(H)\}.
\end{align*}
We learn from this ``elementary'' analysis that
\begin{itemize}
\item $\check{H}(V,A)$ has at least $|\mathcal{W}(H)|$ and at most $R\cdot d_\mathcal{V}$ energy-$0$ edge modes.
\item ``Generically'', it has exactly $|\mathcal{W}(H)|$ ``robust'' energy-$0$ edge states, all lying in one graded component or the other.
\end{itemize}
Generally, the dimension of $\ker \check{H}_{+-}$ and/or $\ker \check{H}_{-+}$ could exceed $|\mathcal{W}(H)|$, and vary wildly with $(V,A)$, so what is the precise formulation of ``generic'' and ``robust''?

\medskip
Remarkably, the number of extra zero-energy states in each graded component is always the same. Thus the weaker statement of Theorem \ref{thm:2.band.BEC} continues to hold even when $R>1$ and/or $d_\mathcal{V}>2$:

\begin{thm}\label{thm:higher.BEC}
Let $H(V,A)$ be a chiral-symmetric gapped Hamiltonian with arbitrary $d_\mathcal{V}$, arbitrary $R$, and $A_R$ not necessarily invertible. Then
\begin{equation}
{\rm Ind}_e(\check{H})\equiv\dim\ker\check{H}_{+-}-\dim\ker\check{H}_{-+}=\mathcal{W}(H).\label{eqn:BEC.many.bands}
\end{equation}
\end{thm}
\begin{proof}
There is an abstract Toeplitz index-theoretic proof of this result (see Section \ref{sec:Toeplitz}). So let us provide the ingredients for an ``elementary'' approach, which is basically a deformation to a direct sum of the three cases in Example \ref{ex:dimerized}. For the latter, $\dim\ker\check{H}_{+-}$ and $\dim\ker\check{H}_{-+}$ are known exactly. Without loss of generality, we assume that $\mathcal{W}(H)\geq 0$.

Recall that $h_{+-}$ is a $\frac{d_\mathcal{V}}{2}\times\frac{d_\mathcal{V}}{2}$ matrix of Laurent polynomials with highest order $\pm R$. Also, $h_{+-}|_{{\rm U}(1)}$ is invertible (gapped assumption), and we call it an invertible \emph{Laurent loop}. Note that an invertible Laurent loop determines the $(V,A)$ of a gapped chiral-symmetric Hamiltonian via the coefficients of its constituent Laurent polynomials; the reverse is true as well. So in what follows, we will find it convenient to talk about deforming the invertible Laurent loop, keeping in mind that this exactly mirrors the physical picture of deforming the gapped chiral-symmetric Hamiltonians.

\begin{itemize}
\item Following Example \ref{ex:dimerized}, define the ``trivial'' Hamiltonian $H_{\rm triv}$ on $\ell^2(\ZZ;\mathcal{V})$ to have $V=\begin{pmatrix} 0 & \mathbf{1}_{d_{\mathcal{V}}/2}\\ \mathbf{1}_{d_{\mathcal{V}}/2} & 0\end{pmatrix}$ and $A_1,\ldots,A_R=0$. Take the direct sum $H(V,A)\oplus H_{\rm triv}$; its Bloch Hamiltonian now has off-diagonal part being $h_{+-}\oplus \mathbf{1}_{d_\mathcal{V}/2}$.

Factorize $h_{+-}(\lambda)=\lambda^{-R}p(\lambda)$, so that $p(\cdot)$ has entries being polynomials of maximum degree $2R$ (no poles). An explicit homotopy
\[
\begin{pmatrix} h_{+-} & 0 \\ 0 & 1_{d_\mathcal{V}/2}\end{pmatrix} \sim\begin{pmatrix}
p & 0 \\ 0 & \lambda^{-R}\cdot \mathbf{1}_{d_\mathcal{V}/2}
\end{pmatrix},
\]
is achieved by
\[
\begin{pmatrix}\cos t & -\sin t \\ \sin t & \cos t\end{pmatrix}
\begin{pmatrix}\lambda^{-R}\cdot\mathbf{1}_{d_\mathcal{V}/2} & 0 \\ 0 & \mathbf{1}_{d_\mathcal{V}/2}\end{pmatrix}
\begin{pmatrix}\cos t & \sin t \\ -\sin t & \cos t\end{pmatrix}
\begin{pmatrix}p & 0 \\ 0 & \mathbf{1}_{d_\mathcal{V}/2}\end{pmatrix},\qquad t\in[0,\pi/2].
\]
Restricted to $\lambda\in {\rm U}(1)$, this gives a homotopy of invertible Laurent loops. The total winding number, $\mathcal{W}(H)$, is maintained throughout, so the \emph{polynomial loop} $p$ alone has winding number $\mathcal{W}(H)+R\cdot {d_\mathcal{V}/2}$.

The loop $\lambda\mapsto \lambda^{-R}\cdot \mathbf{1}_{d_{\mathcal{V}}/2}$ in the second direct summand is easy to ``split up'' into a direct sum of $R\cdot d_{\mathcal{V}}/2$ copies of $\lambda\mapsto\lambda^{-1}$ by a similar procedure.

\item The next step is to add another auxiliary summand $\mathcal{V}^\prime=\mathcal{V}^\prime_+\oplus\mathcal{V}^\prime_-$ of sufficiently large dimension, and introduce another $H_{\rm triv}$ on $\ell^2(\ZZ;\mathcal{V}^\prime)$. The extra space provided by $\mathcal{V}^\prime_+$ is to allow $p\oplus\mathbf{1}_{\mathcal{V}^\prime_+}$ to be homotoped (through invertible polynomial loops) into a matrix $\ell$ of \emph{linear} loops, i.e.\ of degree at most 1. There are standard linear algebraic ways to achieve this, e.g., Prop.\ 2.6 of \cite{Hatcher}, Lemma 9.2.6 of \cite{WO}.

\item The invertible linear loop $\ell$ can be further homotoped (within invertible linear loops) into a direct sum of $(\mathcal{W}(H)+R\cdot {d_\mathcal{V}/2})$ copies of the loop $\lambda\mapsto \lambda$, supplemented by $1$s along the diagonal. This final loop is called a \emph{projection loop}, since it is of the form $\lambda Q +(\mathbf{1}-Q)$ for some projection matrix $Q$. Such a homotopy is described in, e.g.\ Lemma 9.2.7 of \cite{WO}, which we reproduce here. We have $\ell(\lambda)=\lambda C +D$ for some matrices $C,D$, with $C+D=\ell(1)$ invertible. So we can invertibly homotope $\ell$ to
\[
\ell^\prime:=(C+D)^{-1}\cdot\ell:\lambda\mapsto \mathbf{1}+C(\lambda-\mathbf{1}).
\]
For every $\lambda\in{\rm U}(1)\setminus\{1\}$, 
\[
\ell^\prime(\lambda)=(1-\lambda)(\frac{1}{1-\lambda}\mathbf{1}-C)\qquad \textrm{is invertible}.
\]
The map $\lambda\mapsto\frac{1}{1-\lambda}$ takes ${\rm U}(1)\setminus\{1\}$ to the line ${\rm Re}(\mu)=\frac{1}{2}$. So the previous line becomes
\begin{align*}
&(1-\lambda)^{-1}\not\in\sigma(C),\qquad &\forall \lambda\in{\rm U}(1)\setminus\{1\},\\
\Longleftrightarrow\; & \mu\not\in\sigma(C),\qquad &\forall \mu\in\CC:{\rm Re}(\mu)=\frac{1}{2}.
\end{align*}
Because the spectrum of $C$ avoids the ${\rm Re}(\mu)=\frac{1}{2}$ line, we can linearly homotope $C$ to a projection $Q$, without the spectrum ever hitting this critical line. Then $\ell^\prime$ is correspondingly invertibly homotoped to $\mathbf{1}+Q(\lambda-1)=\lambda Q+(\mathbf{1}-Q)$.
\end{itemize}

After adding $H_{\rm triv}$ and performing these homotopies, the final Hamiltonian $H^\prime$ has $h_{+-}^\prime$ being a simple diagonal expression,
\[
h_{+-}^\prime(\lambda)={\rm diag}(\underbrace{\lambda,\ldots,\lambda}_{\mathcal{W}(H)+R\cdot d_\mathcal{V}/2}, \underbrace{\lambda^{-1},\ldots,\lambda^{-1}}_{R\cdot d_\mathcal{V}/2},1,\ldots,1).
\]
Thus $H^\prime$ is just a direct sum of the basic ``dimerized'' examples in Example \ref{ex:dimerized}. Specifically, there are $\mathcal{W}(H)+R\cdot {d_\mathcal{V}/2}$ copies of the $+1$ winding model, $R\cdot {d_\mathcal{V}/2}$ copies of the winding number $-1$ model, and the rest are copies of the winding $0$ model. It follows that
\begin{align*}
\dim \ker \check{H}^\prime_{+-}&=\mathcal{W}(H)+R\cdot {d_\mathcal{V}/2},\\
\dim \ker \check{H}^\prime_{-+}&=R\cdot {d_\mathcal{V}/2},
\end{align*}
and therefore ${\rm Ind}_e(\check{H}^\prime)=\mathcal{W}(H)$.

The spectrum of chiral-symmetric self-adjoint operators is always symmetric about 0. So reversing the above homotopies will at worst result in extra pairs of $\pm E$ eigenvalues being continuously introduced to the zero energy level, which never changes ${\rm Ind}_e(\cdot)$. (Unlike the $d_\mathcal{V}=2$ case, there is no confinement of the discrete spectrum to zero energy.) Here, the stability of discrete spectra is used, \cite{Kato} Chapter IV \S3.5. Also, the extra $\check{H}_{\rm triv}$ do not contribute any edge states at all. We conclude that
\[
{\rm Ind}_e(\check{H})={\rm Ind}_e(\check{H}^\prime)=\mathcal{W}(H). 
\]
\end{proof}

\begin{rem}The reader familiar with $K$-theory and index theory may recognize the above stabilization-homotopy argument as one of the ``hands-on'' steps in the proof of Bott periodicity introduced in \cite{AB}, see \cite{Hatcher,WO} for expositions. To get the actual Bott periodicity result, which is a statement about general loops, spheres, etc.\ in ${\rm GL}(d)$ for large $d$ (not just Laurent loops), some extra work with Fourier analysis is required. 
\end{rem}

Returning to the $d_\mathcal{V}=2$ case, we can improve Theorem \ref{thm:2.band.BEC}.
\begin{thm}\label{thm:2.band.long.range.BEC}
Let $H(V,A)$ be a chiral-symmetric gapped Hamiltonian with $d_\mathcal{V}=2$, and arbitrary $R$. Then $\mathcal{W}(H)\in\{-R,\ldots,R\}$, and
\[
\begin{cases}
\dim\ker\check{H}_{+-}=\mathcal{W}(H),\;\;\;\,\dim\ker\check{H}_{-+}=0,\quad & \mathcal{W}(H)\geq 0,\\
\dim\ker\check{H}_{+-}=0,\qquad\;\;\;\;\,\dim\ker\check{H}_{-+}=-\mathcal{W}(H),\quad & \mathcal{W}(H)\leq 0.
\end{cases}
\]
\end{thm}
\begin{proof}
The off-diagonal part $\check{H}_{+-}$ is a Toeplitz operator (see Section \ref{sec:Toeplitz}) with \emph{scalar}-valued symbol function $h_{+-}$ having winding number $\mathcal{W}(H)$. By Theorem \ref{thm:higher.BEC},
\[
\dim\ker\check{H}_{+-}-\dim\ker\check{H}_{-+}=\mathcal{W}(H).
\]
Then the claim follows from the following Toeplitz operator theory result of \cite{Coburn},
\[
\dim\ker\check{H}_{+-}=0\quad{\rm and/or}\quad\dim\ker\check{H}_{-+}=0.
\]
\end{proof}
Thus, for two-band models, $d_\mathcal{V}=2$, there are always exactly $|\mathcal{W}(H)|\leq R$ edge modes at $0$-energy, regardless of the hopping range $R$. However, when $R\geq 2$, the possibility of edge states with non-zero in-gap energy is not excluded (unlike Prop.\ \ref{prop:edge.modes.energy.zero}).

\subsection{Toeplitz index theory discussion}\label{sec:Toeplitz}
The off-diagonal operator $\check{H}_{+-}:\ell^2(\NN;\mathcal{V}_+)\to\ell^2(\NN;\mathcal{V}_-)$ is an example of a \emph{Toeplitz operator}. The ${\rm End}(\mathcal{V}_+,\mathcal{V}_-)$-valued \emph{symbol} of this Toeplitz operator is precisely the off-diagonal part of the Bloch Hamiltonian function, $h_{+-}|_{{\rm U}(1)}$.
The Toeplitz index theorem (\cite{Noether}, \cite{Arveson} Theorem 4.4.3) reads
\begin{equation}
\underbrace{\textrm{Fredholm index}(\check{H}_{+-})}_{{\rm Ind}_e(\check{H})}\equiv \dim\ker\check{H}_{+-}-\dim\ker\underbrace{\check{H}_{+-}^*}_{=\check{H}_{-+}}=\underbrace{{\rm Wind}(\det h_{+-}|_{{\rm U}(1)})}_{\equiv\mathcal{W}(H)},\label{eqn:Toeplitz.index.theorem}
\end{equation}
which is precisely Eq.\ \eqref{eqn:BEC.many.bands}. Here, there is implicit use of the Fredholm property, that $0$ is avoided in the essential spectrum of $\check{H}_{+-}$ (see \cite{Arveson} Chapter 4.3), in order that the left-hand-side is well-defined.

Eq.\ \eqref{eqn:Toeplitz.index.theorem} holds for Toeplitz operators whose invertible symbol functions only need to be continuous, not necessarily analytic/smooth. The winding number then refers to the homotopy/covering space definition, generalizing the complex analytic one. Hamiltonians which are not strictly finite-range, but are norm-approximated by finite-range ones, have continuous $h_{+-}$. In this sense, the Toeplitz index theorem subsumes Theorem \ref{thm:higher.BEC}, as is known in the mathematical physics literature, e.g.\ Chapter 1.2 of \cite{PSB}, Sec.\ 3B of \cite{TZ}. 

Another technical advantage of the Toeplitz index method is that there is no need to stabilize in the sense of adding an extra $\mathcal{V}^\prime$ to $\mathcal{V}$. This is because $\pi_1({\rm GL}(d))\cong\ZZ$ for all dimensions $d$, via the determinant map. (Note: this does not require the strength of the Bott periodicity of $\pi_n({\rm GL}(\infty))$.) So given any $H(V,A)$ with maximal range $R$ and fixed $d_\mathcal{V}$, its $h_{+-}|_{{\rm U}(1)}$ would be homotopic within continuous ${\rm GL}(d_\mathcal{V})$-valued symbols, to the Hamiltonian with $h_{+-}^\prime(\lambda)={\rm diag}(\lambda^{\mathcal{W}(H)},1,\ldots,1)$. The latter Hamiltonian is basically a dimerized 2-band model, whose half-space version is easily seen to have exactly $|\mathcal{W}(H)|$ edge modes. The catch is that we generally have to ``borrow'' longer range terms, $A_r, r>R$, to achieve this homotopy, since general continuous loops cannot be expressed as Laurent ones with finite expansion. This should be compared to the proof of Theorem \ref{thm:higher.BEC}, where $R$ is never exceeded, but the extra $\mathcal{V}^\prime$ may need to be very large.

\section*{Acknowledgments}
The author thanks H.\ Zhang for stimulating discussions on discrete interface models, which led to a rethinking of the 1D BEC.

\section*{Author declarations}
The author has no conflicts to disclose.

\section*{Data availability statement}
Data sharing is not applicable to this article as no new data were created or analyzed in this study. 

\appendix
\section{Appendix}
\subsection{Validity of exponential ansatz for normal modes}\label{sec:exponential.validity}
Generally, the companion matrix $C_E$ to $H(V,A)$ may not be diagonalizable, and we must consider its generalized eigenvectors as $\mathcal{D}$-valued initial data for the energy-$E$ modes,
\[
v_i\not\in\ker(C_E-\lambda_i),\qquad v_i\in\ker(C_E-\lambda_i)^{m_i},\qquad m_i\geq 2.
\]
In the more familiar case of \emph{scalar-valued} (i.e.\ $\mathcal{V}=\CC$) recurrence equations, one derives, for multiplicity $m=2$ say, that a generalized eigenvector $v_i$ has the form
\[
v_i=(\lambda_i,2\lambda_i^2,3\lambda_i^3,\ldots,2R\cdot\lambda_i^{2R}),
\]
up to adding some genuine eigenvector. This $v_i$ generates the familiar ``polynomial-exponential ansatz'' for the normal modes,
\begin{equation}
\psi_n=n\lambda_i^n,\qquad n\in\ZZ.\label{eqn:poly.exp.ansatz}
\end{equation}
Notice that this $\psi$ is able to satisfy $\psi_0=0$, whereas a normal mode generated from a genuine eigenvector, Eq.\ \eqref{eqn:basic.ansatz}, is always nowhere-vanishing.

For our $\mathcal{V}$-valued case, a generalized eigenvector $v_i\in\mathcal{D}_i^E$ need not have the nice form
\[
v_i=(0,\lambda_i u_i,2\lambda_i^2 u_i,\ldots, 2R\cdot\lambda_i^{2R}u_i),\qquad u_i\in \ker(\mathbf{H}(\lambda_i)-E).\qquad \textrm{(fails)}
\]
This means that the ``polynomial-exponential ansatz'' analogous to Eq.\ \eqref{eqn:poly.exp.ansatz} may not work,
\begin{equation}
\psi_n=n\lambda_i^nu_i,\qquad u_i\in \ker(\mathbf{H}(\lambda_i)-E).\qquad \textrm{(fails)}\label{eqn:poly.exp.fail}
\end{equation}
Nevertheless, we can still compute
\begin{align*}
C_E^{2R}v_i&=(C_E-\lambda_i+\lambda_i)^{2R}v_i\\
&=\underbrace{(C_E-\lambda_i)^{2R}v_i}_{0}+\binom{2R}{1}\lambda_i(C_E-\lambda_i)^{2R-1}v_i+\ldots+\lambda_i^{2R}v_i\\
&=0+\ldots+0+\binom{2R}{m_i-1}\lambda_i^{2R-m_i+1}\underbrace{(C_E-\lambda_i)^{m_i-1}v_i}_{\textrm{eigenvector}}+\ldots+\lambda_i^{2R}v_i.
\end{align*}
This says that after $2R$ advancements by $C_E$, the initial data $v_i\in\mathcal{D}_i^E$ is turned into the above combination of (linearly independent) vectors in the Jordan chain generated by $v_i$. It follows that the normal mode generated by $v_i$ will still decay/blow-up as $n\to\infty$, according to whether $|\lambda_i|<1$ or $|\lambda_i|>1$. The decay/blow-up rate is $\lambda_i$ per-unit-cell, up to polynomial correction factors.

\subsection{Edge states which are not built from exponentials}\label{sec:non.exponential}
For $d_\mathcal{V}=2, R=1$, consider the 1-parameter family of chiral-symmetric models,
\[
V=\begin{pmatrix}0 & 1 \\ 1 & 0\end{pmatrix},\qquad A^{(\theta)}=A^{(\theta)}_1=e^{i\theta}\begin{pmatrix}0 & \frac{1}{4} \\ 1 & 0\end{pmatrix},\qquad \theta\in [-\pi,\pi].
\]
The Bloch Hamiltonians are
\[
\mathbf{H}(\lambda;V,A^{(\theta)})=\begin{pmatrix}
0 & e^{-i\theta}\lambda^{-1}+1+\frac{1}{4}e^{i\theta}\lambda \\
\frac{1}{4}e^{-i\theta}\lambda^{-1}+1+e^{i\theta}\lambda  & 0
\end{pmatrix}.
\]
To calculate the winding number of the off-diagonal element $h_{+-}|_{{\rm U}(1)}$, reparametrize by $\tilde{\lambda}=e^{i\theta}\lambda$, so that $h_{+-}:\tilde{\lambda}\mapsto \frac{1}{4}\tilde{\lambda}^{-1}+1+\tilde{\lambda}$. As $\tilde{\lambda}$ is varied in ${\rm U}(1)$, $h_{+-}$ traces out an ellipse with winding number $+1$ (around the origin). So, by Theorem \ref{thm:2.band.BEC}, there should be one zero-energy edge state for $\check{H}(V,A^{(\theta)})$. Let us calculate this edge state.

The characteristic polynomial for $C_0$ is
\begin{align*}
\det(\lambda-C_0)=\lambda^2\det \mathbf{H}(\lambda)&=(\frac{1}{4}e^{-i\theta}+\lambda+e^{i\theta}\lambda^2)(e^{-i\theta}+\lambda+\frac{1}{4}e^{i\theta}\lambda^2)\\
&=\frac{1}{4}\cdot e^{i2\theta}(\lambda+\frac{1}{2}e^{-i\theta})^2(\lambda+2e^{-i\theta})^2,
\end{align*}
which has the repeated root $-\frac{1}{2}e^{-i\theta}$ inside the unit circle. It is easy to check that $(0,0,1,0)$ is a generalized eigenvector of $C_0(V,A^{(\theta)})$,
\begin{align*}
\left(C_0+\frac{1}{2}e^{-i\theta}\right)\begin{pmatrix} 0 \\ 0 \\ 1 \\ 0\end{pmatrix}
&=\begin{pmatrix}
\frac{1}{2}e^{-i\theta} & 0 & 1 & 0 \\ 
0 & \frac{1}{2}e^{-i\theta} & 0 & 1 \\
-\frac{1}{4} e^{-i2\theta} & 0 & -\frac{1}{2}e^{-i\theta} & 0 \\
0 & -4e^{-i2\theta} & 0 & \frac{9}{2}e^{-i\theta}
\end{pmatrix}
\begin{pmatrix} 0 \\ 0 \\ 1 \\ 0\end{pmatrix}=\begin{pmatrix} 1 \\ 0 \\ -\frac{1}{2}e^{-i\theta} \\ 0 \end{pmatrix},\\
\left(C_0+\frac{1}{2}e^{-i\theta}\right)^2\begin{pmatrix} 0 \\ 0 \\ 1 \\ 0\end{pmatrix}&=0.
\end{align*}
This generalized eigenvector is Dirichlet in the zeroth-unit cell, and it provides the initial data $(\check{\psi}_0, \check{\psi}_1)=(0,0,1,0)$ for the following state,
\begin{equation}
\check{\psi}=(\check{\psi}_1,\check{\psi}_2,\ldots)=\left(\binom{1}{0},\binom{-e^{-i\theta}}{0},\binom{\frac{3}{4}e^{-i2\theta}}{0},\binom{-\frac{1}{2}e^{-i3\theta}}{0},\ldots\right).\label{eqn:poly.exp.example}
\end{equation}
Note that $\check{\psi}$ is \emph{not} a linear combination of exponential ans\"{a}tze states, Eq.\ \eqref{eqn:basic.ansatz}. Yet, it is readily verified that $\check{\psi}$ is precisely the zero-energy edge state of $\check{H}(V,A^{(\theta)})$,
\begin{align*}
(A^{(\theta)})^*\check{\psi}_1+V\check{\psi}_2+A^{(\theta)}\check{\psi}_3&=0,\\
V\check{\psi}_1+A^{(\theta)}\check{\psi}_2&=0.
\end{align*}

We may transform this family of examples by ${\rm U}(1)\times {\rm U}(1)$ unitaries, and also scale them. So there is at least a 4-parameter space of gapped chiral-symmetric Hamiltonians for which the ansatz, Eq.\ \eqref{eqn:basic.ansatz} will fail to give the edge state. 

The point is that the submanifold of those $(V,A)$ for which $C_E(V,A)$ has non-simple eigenvalues, is generally not just a set of points, and is not automatically negligible. In the current case, the full parameter space is sufficiently large:
\[
V=\begin{pmatrix} 0 & \overline{v}\\ v & 0\end{pmatrix},\qquad A=\begin{pmatrix} 0 & a_{-+} \\ a_{+-} & 0 \end{pmatrix},\qquad v,a_{+-},a_{-+}\in\CC,
\]
subject to the gapped constraint. Then the submanifold of ``bad'' $(V,A)$ does not introduce disconnections.

A final remark is that the edge state in these examples, Eq.\ \eqref{eqn:poly.exp.example}, happen to have polynomial-exponential form. It can be viewed as a certain limiting case, $\lambda_1,\lambda_2\to-\frac{1}{2}e^{-i\theta}$, as $E\to 0$. As mentioned in Eq.\ \eqref{eqn:poly.exp.fail}, the story is more complicated in general models, with more $\lambda_i$ and possibly higher multiplicities, and we cannot guarantee the polynomial-exponential ansatz.

\end{document}